\documentclass[letterpaper, 10 pt, conference]{ieeeconf}  

\IEEEoverridecommandlockouts                             

\makeatletter
\let\proof\relax
\let\endproof\relax
\makeatother

\usepackage{graphicx,xcolor} 
\usepackage{amsmath,amsfonts,amsthm}
\usepackage[caption=false,font=footnotesize]{subfig}
\newtheorem{thm}{Theorem}
\newtheorem{lem}{Lemma}
\newtheorem{prop}{Proposition}
\newtheorem{rem}{Remark}
\usepackage{cite}

\usepackage{hyperref}
\hypersetup{colorlinks=true,linkcolor=black,filecolor=black,urlcolor=black,citecolor=black}
\newcommand{\V}{\mathcal{V}}
\newcommand{\Z}{\mathbb{Z}_{\geq 0}}

\title{\LARGE \bf
On the Convergence of an Opinion--Action Coevolution Model with Bounded Confidence
}

\author{Chen Song$^{1,2}$, Angela Fontan$^{2}$, Rong Su$^{1}$, Julien M. Hendrickx$^{3}$, \\ Vladimir Cvetkovic$^{4}$, Karl H. Johansson$^{2}$
\thanks{This work was supported in part by the National Research Foundation, Singapore through its Medium Sized Center for Advanced Robotics Technology Innovation (CARTIN) under Project WP2.7, and in part by the project ``Humanizing the Sustainable Smart City (HiSS)" within the KTH Digital Futures Research Program.}
\thanks{$^{1}$Chen Song and Rong Su are with the School of Electrical and Electronic Engineering, Nanyang Technological University, Singapore, 639798. Email: song0249@e.ntu.edu.sg; rsu@ntu.edu.sg}%
\thanks{$^{2}$Chen Song, Angela Fontan, and Karl H. Johansson are with the Division of Decision and Control Systems, School of Electrical Engineering and Computer Science, KTH Royal Institute of Technology, 114 28 Stockholm, Sweden. Email: \{chensong, angfon, kallej\}@kth.se}%
\thanks{$^{3}$Julien M. Hendrickx is with the ICTEAM Institute, UCLouvain, B-1348, Louvain-la-Neuve, Belgium. His work is supported by the ``SIDDARTA'' Concerted Research Action (ARC). Email: julien.hendrickx@uclouvain.be}%
\thanks{$^{4}$Vladimir Cvetkovic is  with the Division of Resources, Energy and Infrastructure, School of Architecture and the Built Environment, KTH Royal Institute of Technology, 114 28 Stockholm, Sweden. Email: vdc@kth.se}%
}

\begin{document}

\maketitle
\thispagestyle{empty}
\pagestyle{empty}

\begin{abstract}
This paper presents a theoretical convergence analysis for an opinion--action coevolution model that integrates the opinion updating rule of the Hegselmann-Krause model with a utility-based decision-making mechanism. The model is reformulated into an augmented state-space representation, where the state matrix induces a time-varying social interaction digraph. The convergence analysis is grounded on two existing theoretical findings that establish convergence for the Hegselmann-Krause type of models and containment control systems with multiple stationary leaders, respectively. Results indicate that, if the structure of the interaction digraph stabilizes within finite time, the model either converges to consensus, where all agents’ opinions and actions reach an identical state, or exhibits clustering, where some opinion nodes act as stationary leaders while the remaining nodes approach the convex hull formed by the leaders. Numerical simulations are then provided to validate the theoretical results.
\end{abstract}

\section{Introduction} \label{sec1}
The study of opinion dynamics has attracted substantial attention from the systems and control community in recent years \cite{ref1, ref2, ref3, ref3.5}. A variety of mathematical models have been developed to capture how individual opinions evolve under the influence of social interactions. The most representative models include the DeGroot model \cite{ref4}, the Friedkin–Johnsen (FJ) model \cite{ref5}, and the bounded-confidence class of models \cite{ref6, ref7, ref8}, exemplified by the Hegselmann–Krause (HK) model \cite{ref8}. In the HK model, each agent updates its opinion by averaging the opinions of those whose values differ from its own by less than a given confidence threshold. The bounded-confidence updating rule reflects the tendency of individuals to interact with others sharing similar opinions, a phenomenon known as `homophily' in social psychology \cite{ref9}. We refer interested readers to \cite{ref2} and \cite{ref3} for a comprehensive overview of bounded confidence opinion dynamics. 

Inspired by the attitude-behavior gap observed in social psychology \cite{ref10}, our earlier work \cite{ref11} proposed an opinion--action coevolution model that integrates the HK-type opinion updating rule with a utility-based decision-making mechanism. The model employs a sequential updating mechanism, where each agent first revises its opinion based on the observed actions of its neighbors and then updates its action according to its revised opinion and the group’s average action. It has been shown in \cite{ref11} that, by adjusting two key parameters, the model manages to capture a range of socio-psychological phenomena, demonstrating its conceptual validity and strong relevance to actual human social behavior. However, the model remains analytically unexplored, and there is a lack of understanding of its dynamical properties. 

In this work, we address this research gap by establishing convergence results for the opinion--action coevolution model. Our main contributions are summarized as follows. First, we reformulate the model into an augmented state-space representation by defining a composite state vector that jointly encapsulates all agents’ opinions and actions. Second, we characterize the structure of the time-varying interaction digraph induced by the state matrix of the augmented discrete-time system. Then, we establish convergence of the system based on its structural properties and two existing theorems: one for HK-type models \cite{ref12} and the other for containment control systems with multiple stationary leaders \cite{ref13}. Under the assumption that the structure of the system’s interaction digraph stabilizes within finite time, we show that the discrete-time system satisfies the conditions of either one of the theorems in \cite{ref12, ref13}, depending on the topology of the interaction digraph. Finally, the overall convergence results for the opinion-action coevolution model are illustrated through numerical simulations.

The remainder of this paper is organized as follows: Section~\ref{sec2} provides technical preliminaries and introduces the HK model. Section~\ref{sec3} reviews the proposed opinion--action coevolution model and reformulates it into an augmented state-space representation. Section~\ref{sec4} establishes the main convergence results. Section~\ref{sec5} presents numerical evidence validating the theoretical results. Finally, Section~\ref{sec6} concludes the study and proposes directions for future work.   

\section{Preliminaries} \label{sec2}
\textit{Notations:} The sets of nonnegative and positive integers are denoted by \(\Z\) and \(\mathbb{Z}_{>0}\), respectively. The cardinality of a set \(\mathcal{S}\) is denoted by \(|\mathcal{S}|\). The convex hull of a set \(\mathcal{S}\) is denoted by \(\operatorname{conv}(\mathcal{S})\). The distance from a point \(x\) to a set \(\mathcal{S}\) is denoted by \(\mathrm{dist}(x,\mathcal{S}):= \inf_{y \in \mathcal{S}} ||x-y||_2\), where \(\inf\) denotes the infimum and \(||\cdot||_2\) denotes the Euclidean 2-norm.

\subsection{Linear Algebraic Preliminaries} \label{sec2.1}

Given a square matrix \( A=[a_{ij}] \in \mathbb{R}^{n \times n} \), the \( (i,j) \)-th entry of \(A\) is denoted as \(a_{ij}\) or \( [A]_{ij} \). All column vectors are denoted by bold lowercase letters, e.g., \( \mathbf{x} \in \mathbb{R}^n \). The transpose of a matrix \( A \) or vector \( \mathbf{x} \) is denoted as \( A^\top \) or \( \mathbf{x}^\top \). The identity matrix of dimension \(n \times n\) is denoted by \( I_n \in \mathbb{R}^{n \times n} \), and the symbol \( \mathbf{1}_n=[1 \, \cdots \, 1]^\top \in \mathbb{R}^n \) denotes the column vector of all ones. A matrix \(A \in \mathbb{R}^{n \times n}\) is nonnegative if \(a_{ij} \geq 0\) for all \(i, j \in \{1,2,\dots,n\}\). A matrix \(A\) is row-stochastic if it is nonnegative and satisfies \(A\mathbf{1}_n=\mathbf{1}_n\), i.e., each row of \(A\) sums to one. A matrix \(A\) is row-substochastic if it is nonnegative and satisfies \(A\mathbf{1}_n \leq \mathbf{1}_n\), with at least one row-sum strictly less than 1.

\subsection{Graphs} \label{sec2.2}

Let \( \mathcal{G}(A)=(\mathcal{V},\mathcal{E},A) \) denote a weighted digraph, where \(\mathcal{V}=\{1,2,\dots,n\}\) is the node index set and \( \mathcal{E} \subseteq \mathcal{V} \times \mathcal{V}\) is the set of directed edges. The nonnegative matrix \(A \in \mathbb{R}^{n \times n}\) is the weighted adjacency matrix corresponding to \(\mathcal{G}(A)\), where \(a_{ij}>0\) if and only if \( (j,i)\in \mathcal{E} \). A self-loop \((i,i)\in \mathcal{E}\) is a directed edge from node \(i\) to itself. A directed walk from node \(i\) to \(j\) is an ordered sequence of edges \((i,i_1),(i_1,i_2),\dots,(i_{s-1},i_s),(i_s,j)\) selected from \(\mathcal{E}\). For simplicity, we denote a digraph by \(\mathcal{G}\) or \(\mathcal{G}(A)\).

A digraph \(\mathcal{G}\) is strongly connected if there exists a directed walk from \(i\) to \(j\) for every pair of nodes \(i,j \in \mathcal{V}\), i.e., every node is reachable from every other node. A digraph \( \mathcal{G}'=(\mathcal{V}',\mathcal{E}') \) is called a subgraph of \( \mathcal{G}=(\mathcal{V},\mathcal{E}) \) if \(\mathcal{V}' \subseteq \mathcal{V}\) and \(\mathcal{E}' \subseteq \mathcal{E} \,\cap\, (\mathcal{V}' \times \mathcal{V}') \). In particular, the subgraph of \( \mathcal{G}=(\mathcal{V},\mathcal{E}) \) induced by \(\mathcal{V}' \subseteq \mathcal{V}\) is the digraph \( \mathcal{G}'=(\mathcal{V}',\mathcal{E}') \), where \(\mathcal{E}'=\{(i,j) \in \mathcal{E} \mid i,j \in \mathcal{V}' \}\). A subgraph \(\mathcal{G}'\) is a strongly connected component (SCC) of \(\mathcal{G}\) if it is strongly connected and any other subgraph of \(\mathcal{G}\) strictly containing \(\mathcal{G}'\) is not strongly connected. A singleton SCC of \(\mathcal{G}\) contains only one node. The condensation digraph of \(\mathcal{G}\), denoted by \(C(\mathcal{G})\), is a digraph whose nodes correspond to the SCCs of \(\mathcal{G}\), and there exists a directed edge from node \(i\) to node \(j\) in \(C(\mathcal{G})\) if and only if there exists a directed edge from a node of SCC\(_i\) to a node of SCC\(_j\) in \(\mathcal{G}\). By definition, \(C(\mathcal{G})\) has  no self-loops. A singleton node of \(C(\mathcal{G})\) corresponds to a singleton SCC of \(\mathcal{G}\). In a digraph \(\mathcal{G}\), a source is a node with no incoming edges, and a sink is a node with no outgoing edges.    

\subsection{The Hegselmann-Krause Model} \label{sec2.3}

Let $\V=\{1,2,\dots,n\}$ denote the index set of \(n\) agents. Each agent \(i \in \V\) holds an opinion \(x_i(t)\in [0,1]\) at time \( t \) regarding a certain issue. In the HK model \cite{ref8}, each agent's opinion is updated as:
\begin{equation} \label{e0}
x_i(t+1) = \frac{1}{|N_i(t)|} \sum_{j \in N_i(t)} x_j(t), \;\; i \in \V,\, t\in \Z,
\end{equation}
where \(N_i(t)\) represents agent \(i\)'s neighbor set, defined as:
\begin{equation} \label{e0+}
N_i(t) = \{ j \in \V \mid |x_i(t) - x_j(t)| \leq \epsilon \}.
\end{equation}
The parameter \(\epsilon \in [0,1]\) indicates the homogeneous confidence threshold, which represents agents' susceptibility to interpersonal social influence. It has been established in \cite{ref14, ref15, ref16} that for every \(i \in \V\),  \(x_i(t)\) converges to a limit \(x_i^\ast\) in finite time, and for any \(i,j \in \V\), either \(x_i^\ast = x_j^\ast\) or \(|x_i^\ast - x_j^\ast| > \epsilon\) holds. In other words, the HK model converges in finite time and the inter-cluster distances at equilibrium are bounded below by \(\epsilon\). 

The following result, as stated in \cite{ref12}, provides a general convergence theorem applicable to the HK-type models and serves as a key tool in the convergence analysis of our model in Section~\ref{sec4}.

\begin{thm}[Theorem~2 in \cite{ref12}] \label{thm1}  
Let \( x : \mathbb{N} \to \mathbb{R}^{n} \) satisfy 
\[
x_i(t+1) = \sum_{j=1}^{n}  a_{ij}(t) x_j(t), \quad i = 1, \dots, n,
\]
where \( a_{ij}(t) \geq 0 \) for all \( i, j, t\), and \( \sum_{j=1}^{n} a_{ij}(t) = 1 \) for all \( i \) and \( t \). Suppose that the following conditions hold:
\begin{enumerate}
\item[(i)] \textit{Lower bound on positive coefficients:} there exists some constant \( \alpha > 0 \) such that if \( a_{ij}(t)>0 \), then \( a_{ij}(t) \geq \alpha \) for all \( i, j,\) and \( t \).  
\item[(ii)] \textit{Positive diagonal coefficients:} there exists some constant \( \beta > 0 \) such that \( a_{ii}(t) \geq \beta \) for all \( i \) and \( t \).
\item[(iii)] \textit{Cut-balance:} for any nonempty proper subset \( S \) of \(\{1, \dots, n\} \), there exist \( i \in S \) and \( j \notin S \) with \( a_{ij}(t)>0 \) if and only if there exist \( i' \in S \) and \( j' \notin S \) with \( a_{j'i'}(t)>0 \). 
\end{enumerate}
Then, the limit \( x_i^\ast = \lim_{t \to \infty} x_i(t) \) exists, \( \forall \, i \in \{1,\dots,n\} \). Furthermore, consider the digraph \(\mathcal{G}=(\mathcal{V},\mathcal{E})\) in which \((j,i) \in \mathcal{E}\) if \(a_{ij}(t)>0\) infinitely often. If \(i\) and \(j\) belong to the same strongly connected component of \(\mathcal{G}\), then \(x_i^\ast = x_j^\ast\).  
\end{thm}

\section{Problem Statement} \label{sec3}

In this section, we first review the opinion--action coevolution model proposed in our earlier work \cite{ref11}, and then reformulate it into an augmented state-space representation to facilitate the subsequent convergence analysis. 

\subsection{Model Description} \label{sec3.1}

We consider a population of \( n \) agents, also indexed by \( \V = \{1, 2, \dots, n\} \). Each agent \( i \in \V \) is characterized by an opinion variable \( x_i \) and action variable \( y_i \), where \( x_i, y_i \in [0,1] \). At each discrete time step \( t \in \Z \), agents update their opinions and actions sequentially. 

The opinion updating rule is given by:
\begin{equation} \label{e1}
x_i(t+1) = \frac{\sum_{j \in \mathcal{N}_i(t)} y_j(t) + x_i(t)}{|\mathcal{N}_i(t)| + 1}, \; i \in \mathcal{V},
\end{equation}
where each agent's neighbor set \( \mathcal{N}_i(t) \) is defined as:
\begin{equation} \label{e2}
\mathcal{N}_i(t) = \{ j \in \mathcal{V} \mid j \neq i, \ |x_i(t) - y_j(t)| \leq \epsilon \},
\end{equation}
where \( \epsilon \in [0,1] \) is the homogeneous confidence threshold parameter as defined in the HK model.
The action updating rule is given by:
\begin{equation} \label{e3}
y_i(t+1) = \phi \, x_i(t+1) + (1-\phi) \, y_{\mathrm{avg}}(t), \; i \in \mathcal{V},
\end{equation}
where $y_{\mathrm{avg}}(t) = \frac{1}{n} \sum_{k=1}^n y_k(t)$ is the group's average action, and \( \phi \in [0,1] \) is the homogeneous decision weight parameter that reflects the extent to which each agent is committed to its own opinion when making a decision.

When \( \phi =1 \), agents are fully committed to their opinions when making a decision, i.e., \(y_i=x_i, \, \forall \, i \in \mathcal{V}\). When \( \phi =0 \), agents completely conform to the group's average action, which can be interpreted as the social norm \cite{ref17}. In reality, \(y_{\mathrm{avg}}(t)\) could be disseminated to each agent via social media.  

\begin{rem} \label{remark0}
\upshape The opinion--action coevolution model assumes that each agent’s opinion is private and known only to themselves, while their actions are public and observable to others. It has been demonstrated in \cite{ref11} that this model could capture and explain several real-world phenomena from a socio-psychological perspective. We refer interested readers to \cite{ref11} for a comprehensive account of our model. It is worth noting that when \( \phi=1 \), our model reduces to the HK model. 
\end{rem}

\subsection{Augmented State-Space Representation} \label{sec3.2}

For the analysis to be presented in Section~\ref{sec4}, it is useful to reformulate the opinion--action coevolution model \eqref{e1}--\eqref{e3} into the following augmented state-space representation:
\begin{equation} \label{eq:opinion-action-z}
\mathbf{z}(t+1) = P(t)\, \mathbf{z}(t), \ \ t \geq 1,
\end{equation}
where the augmented state vector is defined as:   
\begin{equation} \label{e6}
\begin{aligned}
\mathbf{z}(t) &= [\, z_1(t), \cdots, z_n(t), z_{n+1}(t), \cdots, z_{2n}(t) \,]^\top \\
       &= [\, x_1(t), \cdots, x_n(t), y_1(t-1), \cdots, y_n(t-1) \,]^\top.
\end{aligned}
\end{equation}
The state vector \( \mathbf{z} \in \mathbb{R}^{2n} \) and state matrix \( P \in \mathbb{R}^{2n \times 2n} \) can be partitioned as follows:
\begin{equation*}
\mathbf{z}(t+1) = 
\begin{bmatrix}
\mathbf{z}_1(t+1) \\[6pt] \mathbf{z}_2(t+1)
\end{bmatrix}
= 
\begin{bmatrix}
P_{11}(t) & P_{12}(t) \\[6pt]
P_{21}(t) & P_{22}(t) 
\end{bmatrix}
\begin{bmatrix}
\mathbf{z}_1(t) \\[6pt] \mathbf{z}_2(t)
\end{bmatrix},
\end{equation*}
where \( \mathbf{z}_1 = [z_1, \cdots, z_n]^\top \) and \( \mathbf{z}_2 = [z_{n+1}, \cdots, z_{2n}]^\top \).

To derive the state matrix \(P(t)\), we proceed as follows. Substituting \eqref{e3} into \eqref{e1}, we obtain: 
\begin{equation*}
\begin{aligned}
&x_i(t+1)  \\[2pt]
&= \frac{x_i(t) + \phi \sum_{j \in \mathcal{N}_i(t)}x_j(t) + (1-\phi) \sum_{j \in \mathcal{N}_i(t)}y_{\mathrm{avg}}(t-1)}{|\mathcal{N}_i(t)|+1} \\[2pt]
&= \frac{x_i(t) + \phi \sum_{j \in \mathcal{N}_i(t)}x_j(t) + \tfrac{1-\phi}{n}\, |\mathcal{N}_i(t)| \sum_{k=1}^n y_k(t-1)}{|\mathcal{N}_i(t)|+1}.
\end{aligned}
\end{equation*}
Using the augmented state \eqref{e6}, for all \(i \in \mathcal{V}\) we rewrite it as: 
\begin{equation*}
\begin{aligned}
&z_i(t+1) \\[2pt]
&= \frac{z_i(t)+\phi \sum_{j \in \mathcal{N}_i(t)}z_j(t)+\tfrac{1-\phi}{n}\, |\mathcal{N}_i(t)| \sum_{k=1}^n z_{n+k}(t)}{|\mathcal{N}_i(t)|+1}.
\end{aligned}
\end{equation*}
Thus, the block matrices \( P_{11}(t) \) and \( P_{12}(t) \) are given by:
\begin{gather}
\big[P_{11}(t)\big]_{ij} = 
\begin{cases}
\tfrac{1}{|\mathcal{N}_i(t)|+1}, & j=i, \\[3pt]
\tfrac{\phi}{|\mathcal{N}_i(t)|+1}, & j \in \mathcal{N}_i(t), \\[3pt]
0, & \text{otherwise}, 
\end{cases}
\;\, i,j=1, \dots, n,
\notag\\[6pt]
\label{e12}
\big[P_{12}(t)\big]_{ij}
= \dfrac{(1-\phi) |\mathcal{N}_i(t)|}{(|\mathcal{N}_i(t)|+1)\,n},
\quad i,j=1, \dots, n.
\end{gather}
In other words, each row \( i \) of \( P_{11}(t) \) consists of a diagonal entry \( \tfrac{1}{|\mathcal{N}_i(t)|+1} \), off-diagonal entries \( \tfrac{\phi}{|\mathcal{N}_i(t)|+1} \) in the neighbor columns, and \( 0 \) elsewhere, while each row \( i \) of \( P_{12}(t) \) contains identical entries equal to \( \tfrac{(1-\phi)|\mathcal{N}_i(t)|}{(|\mathcal{N}_i(t)|+1)\,n} \). 

Substituting $y_{\mathrm{avg}}(t)$ into \eqref{e3}, we obtain:
\begin{equation*}
y_i(t) = \phi \, x_i(t) + \frac{1-\phi}{n} \sum_{k=1}^n y_k(t-1).
\end{equation*}
Using the augmented state \eqref{e6}, for all \(i \in \mathcal{V}\) we rewrite it as: 
\begin{equation*}
z_{n+i}(t+1) = \phi \, z_i(t) + \frac{1-\phi}{n} \sum_{k=1}^n z_{n+k}(t).
\end{equation*}
Thus, the block matrices \( P_{21}(t) \) and \( P_{22}(t) \) are given by:
\begin{gather}
\label{e15}
P_{21}(t)  = \phi I_n,
\quad 
P_{22}(t) = \tfrac{1-\phi}{n} \, \mathbf{1}_n\mathbf{1}_n^\top.
\end{gather} 
Note that the submatrices \( P_{21} \) and \( P_{22} \) are time-invariant, whereas \( P_{11} \) and \( P_{12} \) change over time as a function of \(\mathbf{z}(t)\) since each agent's neighbor set \(\mathcal{N}_i(t)\) may be time-varying. 

\section{Convergence analysis} \label{sec4}
In this section, we establish the convergence of the opinion--action coevolution model as shown in \eqref{eq:opinion-action-z}. We first examine the state matrix \(P(t)\), as defined in Section~\ref{sec3.2}, to verify whether it satisfies the conditions of Theorem~\ref{thm1} and to analyze its structural properties. Then we present the main convergence results based on these properties.  

\subsection{Properties of the State Matrix \(P(t)\)}
We first show that the state matrix \(P(t)\) is always row-stochastic (Lemma~\ref{lem1}). Then, we verify that \( P(t) \) satisfies conditions~(i) and~(ii) of Theorem~\ref{thm1} when \( \phi\in(0,1) \) (Lemmas~\ref{lem2} and~\ref{lem3}). Next, we propose a sufficient condition under which \( P(t) \) satisfies condition~(iii) of Theorem~\ref{thm1} (Proposition~\ref{prop1}). Lastly, we establish a key result on the connectivity of the weighted digraph associated with \( P(t) \) when \( \phi\in(0,1) \) (Theorem~\ref{thm2}).

\begin{lem} \label{lem1}
For all \( t\geq 1 \), \( P(t) \) is a row-stochastic matrix.  
\end{lem}

\proof
From the block matrices' definitions in \eqref{e12} and \eqref{e15}, it is apparent that all entries of \( P(t)=\big[p_{ij}(t)\big] \in \mathbb{R}^{2n \times 2n} \) are nonnegative for all \( t\geq 1\). 

For each row \( i \in \{1,2,\dots,n\} \) in the upper blocks \( \big[P_{11}(t), P_{12}(t)\big]  \), we have: 
\[
\sum_{j=1}^{2n} p_{ij}(t)
=\frac{1+\phi|\mathcal{N}_i(t)|}{|\mathcal{N}_i(t)|+1}+\frac{(1-\phi)|\mathcal{N}_i(t)|}{|\mathcal{N}_i(t)|+1}=1,
\]
while for each row \( i \in \{n+1,n+2,\dots,2n\}\) in the lower blocks \( \big[P_{21}(t), P_{22}(t)\big]  \), the row-sum is equal to:
\[
\sum_{j=1}^{2n} p_{ij}(t)
=\phi+\frac{(1-\phi)n}{n}=1.
\]

\noindent In other words, \(P(t)\) is nonnegative and each row sums to one. Thus, \( P(t) \) is row-stochastic for all \( t\geq 1 \).
\endproof

\begin{lem} \label{lem2}
When \(\phi\in (0,1)\), there exists a constant \( \alpha>0 \) such that if \( p_{ij}(t)>0 \), then \( p_{ij}(t) \geq \alpha\) for all \( i, j \in \{1,2,\dots,2n\}\) and \(t \geq 1 \). 
\end{lem}

\proof
When \(\phi\in (0,1)\), consider each block matrix:

\begin{itemize}
\item \( P_{11}(t) \) has strictly positive entries equal to \( \tfrac{1}{|\mathcal{N}_i(t)|+1} \) and \( \tfrac{\phi}{|\mathcal{N}_i(t)|+1} \), which are bounded below by \( \tfrac{\phi}{n} \) since \(\phi<1\), and \( |\mathcal{N}_i(t)| \leq n-1 \) for all \( i\in \mathcal{V} \).
\item If \( |\mathcal{N}_i(t)|>0\) for some \(i \in \mathcal{V} \), \( P_{12}(t) \) has strictly positive entries equal to \( \tfrac{(1-\phi)|\mathcal{N}_i(t)|}{(|\mathcal{N}_i(t)|+1)n} \), which are bounded below by \( \tfrac{1-\phi}{2n} \), achieved when \(|\mathcal{N}_i(t)|=1\).
\item \( P_{21}(t) \) has strictly positive entries equal to \( \phi \), which serves as the positive lower bound itself.
\item \( P_{22}(t) \) has strictly positive entries equal to \( \tfrac{1-\phi}{n} \), which serves as the positive lower bound itself.  
\end{itemize}
Thus, every strictly positive entry of \(P(t)\) is bounded below by \(\alpha = \min\{\tfrac{\phi}{n},\,\tfrac{1-\phi}{2n}, \phi,\tfrac{1-\phi}{n}\}>0\) for all \(t \geq 1\).    
\endproof

\begin{lem} \label{lem3}
When \(\phi\in (0,1)\), there exists a constant \( \beta>0 \) such that \(p_{ii}(t) \geq \beta\) for all \(i \in \{1,2,\dots,2n\}\) and \(t\geq1\).
\end{lem}

\proof
For \( i\in\{1,\dots,n\} \), \(p_{ii}(t)=\tfrac{1}{|\mathcal{N}_i(t)|+1} \geq \tfrac{1}{n}\). For \( i\in\{n+1,\dots,2n\} \), \(p_{ii}(t)=\tfrac{1-\phi}{n}>0\). Hence, every diagonal entry of \( P(t) \) is bounded below by \( \beta=\min\{\,\tfrac{1}{n},\,\tfrac{1-\phi}{n}\,\}=\tfrac{1-\phi}{n}>0\) for all \(t \geq 1\).
\endproof

It remains to examine whether the state matrix \( P(t) \) satisfies condition~(iii) of Theorem~\ref{thm1}, namely the cut-balance property, which will be analyzed from a graph-theoretic perspective. The following proposition provides a sufficient condition for a digraph to be cut-balanced.   

\begin{prop} \label{prop1}
If a digraph is strongly connected, then it is cut-balanced. 
\end{prop}

\proof
Let \( \mathcal{G}(A)=(\mathcal{V},\mathcal{E},A) \) be a strongly connected digraph with \(n\) nodes. For any nonempty proper subset \( S \subset \mathcal{V}=\{1,\dots,n\}\), its complement is denoted \( S^c = \mathcal{V} \setminus S \). Assume \( u \in S \) and \( v \in S^c \). Since \( \mathcal{G}(A) \) is strongly connected, there exists a directed walk from node \( u \) to node \( v \). This walk must include an edge \( (i',j') \in \mathcal{E} \), i.e., \( a_{j'i'}>0 \), where \( i' \in S \) and \( j' \in S^c \). Similarly, there also exists a directed walk from node \( v \) to node \( u \), which must contain an edge \( (j,i) \in \mathcal{E} \), i.e., \( a_{ij}>0 \), where \( i \in S \) and \( j \in S^c \). Therefore, edges exist in both directions across any cut of \( \mathcal{G}(A)\), which proves the cut-balance property.
\endproof

The following theorem establishes a key structural property of the weighted digraph associated with \( P(t) \).

\begin{thm} \label{thm2}
When \( \phi \in (0,1) \), the weighted digraph associated with \( P(t) \), denoted by \( \mathcal{G}\big(P(t)\big) \), has one of the following structures at any time \( t \geq 1 \):
\begin{enumerate}
\item[(i)] it is strongly connected;
\item[(ii)] its condensation digraph consists of one sink and multiple singleton sources.     
\end{enumerate}
\end{thm}

\proof
Let \( \mathcal{V}_P = \{1,2,\dots,2n\} \) denote the vertex set of \( \mathcal{G}\big(P(t)\big) \). By definition of state variables in \eqref{e6}, at any time step \( t \geq 1 \), the indices \( \{1, \dots, n\} \) correspond to the opinion nodes with values equal to \( \{ x_1(t),\dots,x_n(t) \} \), and the indices \( \{n+1, \dots, 2n\} \) correspond to the action nodes with values equal to \( \{y_1(t-1),\dots,y_n(t-1)\} \). 

From \eqref{e15}, every entry of \( P_{22}(t) \) is equal to \(\tfrac{1-\phi}{n}\), which is strictly positive when \(\phi \in (0,1)\). This implies that there is a directed edge from every action node to every other action node in \( \mathcal{G}\big(P(t)\big)\), i.e., every action node is reachable from every other action node. Similarly, \( P_{21}(t)=\phi I_n\) has positive diagonal entries when \( \phi\in (0,1) \), implying that there is a directed edge from each opinion node indexed by \( i \in \{1,\dots,n\} \) to its corresponding action node indexed by \( n+i \) in \( \mathcal{G}\big(P(t)\big)\). As a result, for any \(i,j \in \{1,\dots,n\}\), opinion node \(i\) can reach action node \(n+j\) through its corresponding action node \(n+i\), since all action nodes are mutually reachable.

Define two subsets of \( \mathcal{V}_P\) as: 
\[
\Omega(t) := \{n+1, \dots, 2n\}  \cup  \{i \in \{1,\dots,n\}\;|\; |\mathcal{N}_i(t)|>0\},
\]
\[
\Theta(t) := \{i \in \{1,\dots,n\}\;|\; |\mathcal{N}_i(t)|=0\},
\]
where \( \Omega(t) \cup \Theta(t) = \mathcal{V}_P \) and \( \Omega(t) \cap \Theta(t) = \emptyset \).

\textbf{Case (i).} If \( \Theta(t) = \emptyset \), or equivalently, \( \Omega(t)=\mathcal{V}_P \), then \( |\mathcal{N}_i(t)|>0 \) for all \( i \in \{1,\dots,n\} \), i.e., each agent has at least one neighbor at time \( t \).


From \eqref{e12}, the \((i,j)\)-th entry of \(P_{12}(t)\) is given by:
\[
\big[P_{12}(t)\big]_{ij}
= \dfrac{(1-\phi) |\mathcal{N}_i(t)|}{(|\mathcal{N}_i(t)|+1)\,n},\, \forall \, i,j \in \{1, \dots, n\},
\]
which is strictly positive when \( \phi \in (0,1) \) and \(|\mathcal{N}_i(t)|>0\). 

Hence, every entry of \(P_{12}(t)\) is strictly positive because \( |\mathcal{N}_i(t)|>0 \) for all \( i \in \{1,\dots,n\} \). This implies that there is a directed edge from every action node to every opinion node in \( \mathcal{G}\big(P(t)\big)\), i.e., every opinion node is reachable from every action node.

Recall that every opinion node can reach every action node. Thus, for any two opinion nodes \(i,j \in \{1,\dots,n\}\), they are also mutually reachable. 

Therefore, in this case, \( \mathcal{G}\big(P(t)\big)\) is strongly connected since every node is reachable from every other node. For an intuitive illustration of the digraph structure, the reader is referred to an example presented in Fig.~\ref{fig1}(b) of Section~\ref{sec5}.

\textbf{Case (ii).} If \( \Theta(t) \neq \emptyset \), or equivalently, \( \Omega(t) \subset \mathcal{V}_P \), then \( \exists \, i \in \{1,\dots,n\}\) such that \(|\mathcal{N}_i(t)|=0\), i.e., at least one agent has no neighbors at time \(t\).

For each \( i \in \Theta(t) \), the adjacency structure of opinion node \(i\) is analyzed as follows. From \eqref{e12}, the \(i\)-th row of \(P_{11}(t)\) corresponds to the \(i\)-th row of \(I_n\), with its diagonal entry equal to 1 and all other entries equal to 0. This implies that opinion node \(i\) admits a self-loop, with no incoming edges from any other opinion node. Similarly, all entries in the \(i\)-th row of \( P_{12}(t) \) are identically zero, indicating that opinion node \(i\) has no incoming edges from any action node. Therefore, each opinion node \( i \in \Theta(t) \) has no incoming edges from any other node, and thus forms a singleton source in the condensation digraph of \( \mathcal{G}\big(P(t)\big)\). As a byproduct, from \eqref{e1}, it follows that \(z_i(t+1)=x_i(t+1)=x_i(t)=z_i(t)\), i.e., the value of node \(i\) remains constant during this iteration. 

In contrast, every opinion node \(j \in \{1,\dots,n\} \setminus \Theta(t)\) remains reachable from every action node because all entries in the \(j\)-th row of \(P_{12}(t)\) are positive when \(|\mathcal{N}_j(t)|>0\). Thus, the subgraph of \( \mathcal{G}\big(P(t)\big)\) induced by \(\Omega(t)\) is still strongly connected, similar to Case~(i), which forms an SCC of \(\mathcal{G}\big(P(t)\big)\) because any other opinion node \( i \in \Theta(t) \) is not reachable from \(\Omega(t)\). The SCC formed by \(\Omega(t)\) serves as a sink in the condensation digraph of \(\mathcal{G}\big(P(t)\big)\) since it has no outgoing edges to any singleton source \(i \in \Theta(t)\).    

As a result, in this case, the condensation digraph of \(\mathcal{G}\big(P(t)\big)\) consists of one sink and multiple singleton sources. An example digraph corresponding to this structure can be found in Fig.~\ref{fig2}(b) of Section~\ref{sec5}.
\endproof

\begin{rem} \label{remark1}
\upshape It should be noted that the condition \( |\mathcal{N}_i(t)|>0 \) for all \( i \in \mathcal{V} \) is both sufficient and necessary for the digraph \( \mathcal{G}\big(P(t)\big)\) to be strongly connected at time \(t\geq1\). Case~(i) of Theorem~\ref{thm2} proves sufficiency, and Case~(ii) establishes necessity since \( \mathcal{G}\big(P(t)\big)\) is not strongly connected if some agent has no neighbors at time \(t\geq 1\).  
\end{rem}   

\subsection{Main Convergence Results}
We are now ready to state the main convergence results for the opinion--action coevolution model. As mentioned in Remark~\ref{remark0}, when \( \phi = 1 \), the model \eqref{e1}-\eqref{e3} reduces to the HK model \cite{ref8}, whose convergence has been well established in \cite{ref14, ref15, ref16}. The following lemma establishes convergence of the opinion--action coevolution model when \(\phi = 0\).

\begin{lem} \label{lem4}
When \(\phi = 0\), all agents reach a consensus in action after one update, i.e., \(y_i(t)=y_\mathrm{avg}(0)\) for all \(i \in \mathcal{V}\) and \(t \geq 1\). Each agent's opinion either remains constant, i.e., \(x_i(t)=x_i(1)\) for all \(t \geq 1\), if \(|x_i(1)-y_\mathrm{avg}(0)| > \epsilon\), or converges exponentially to \(y_\mathrm{avg}(0)\), i.e., \(\lim_{t \to \infty} x_i(t) = y_\mathrm{avg}(0)\), if \(|x_i(1)-y_\mathrm{avg}(0)| \leq \epsilon\). 
\end{lem} 

\proof
When \(\phi=0\), from \eqref{e3}, \(y_i(t+1)=y_\mathrm{avg}(t)\) for all \(i \in \mathcal{V}\) and \(t \geq 0\), which leads to \(y_\mathrm{avg}(t+1)=y_\mathrm{avg}(t)\) for all \(t \geq 0\). Thus, all agents' actions reach a consensus at time \(t=1\) and remain constant afterwards, i.e., \(y_i(t)=y_\mathrm{avg}(0), \, \forall \, i \in \mathcal{V}, \, \forall \, t \geq 1\).

From \eqref{e2}, if \(|x_i(1)-y_\mathrm{avg}(0)| > \epsilon\), then agent \(i\) has no neighbors at time \(t=1\). In this case, from \eqref{e1}, \(x_i(2)=x_i(1)\), which leads to \(|x_i(2)-y_\mathrm{avg}(0)| > \epsilon\) and \(x_i(3)=x_i(2)\). Hence, agent \(i\) always has no neighbors and its opinion remains constant for all \(t\geq1\), i.e., \(x_i(t)=x_i(1),\, \forall \, t \geq 1\). 

Conversely, if \(|x_i(1)-y_\mathrm{avg}(0)| \leq \epsilon\), then agent \(i\) always has \(n-1\) neighbors with identical action \(y_\mathrm{avg}(0)\) for all \(t\geq1\). In this case, its opinion updating rule \eqref{e1} can be rewritten as: \(x_i(t+1)-y_\mathrm{avg}(0)=\tfrac{1}{n}(x_i(t)-y_\mathrm{avg}(0)), \, \forall \, t\geq1\), implying that \(x_i(t)\) converges exponentially to \(y_\mathrm{avg}(0)\) with ratio \(\tfrac{1}{n}\). Thus, \(\lim_{t \to \infty} x_i(t) = y_\mathrm{avg}(0)\).  
\endproof

For \(0<\phi<1\), the digraph \(\mathcal{G}\big(P(t)\big)\) necessarily belongs to one of the two structural cases presented in Theorem~\ref{thm2} at any time step \(t \geq 1\). We analyze each of the two cases in Propositions~\ref{prop2} and~\ref{prop3}, respectively, and summarize the overall result in Theorem~\ref{thm3}. The following proposition proves convergence for Case~(i) of Theorem~\ref{thm2}.

\begin{prop} \label{prop2}
Suppose that there exists a finite time instant \(T > 0\) such that for all \(t \geq T\), the weighted digraph \(\mathcal{G}\big(P(t)\big)\) remains strongly connected. Then the system~\eqref{eq:opinion-action-z} converges to consensus, i.e., \(z_i^\ast = z_j^\ast, \, \forall \, i,j \in \{1,\dots,2n\}\). 
\end{prop}

\proof
Assume that \(\mathcal{G}\big(P(t)\big)\) is strongly connected for all \(t \geq T\). From Proposition~\ref{prop1}, a strongly connected digraph is cut-balanced, implying that the state matrix \(P(t)\) satisfies condition~(iii) of Theorem~\ref{thm1} infinitely often. By Theorem~\ref{thm1}, the limit \( z_i^\ast = \lim_{t \to \infty} z_i(t) \) exists for all \( i \in \{1,\dots,2n\} \). 

Moreover, consider the limiting digraph \(\mathcal{G}_P=(\mathcal{V}_P,\mathcal{E}_P)\), as defined in Theorem~\ref{thm1}, where \(\mathcal{V}_P=\{1,\dots,2n\}\) and \( (j,i) \in \mathcal{E}_P \) if \(p_{ij}(t)>0\) infinitely often. Since \(P_{21}\) and \(P_{22}\) are time-invariant, and \(\mathcal{G}\big(P(t)\big)\) is strongly connected for all \(t \geq T\), implying that every entry of \(P_{12}(t)\) is positive infinitely often (see Remark~\ref{remark1}), \(\mathcal{G}_P\) is also strongly connected. According to Theorem~\ref{thm1}, all states reach a consensus at convergence, i.e., \(z_i^\ast = z_j^\ast, \, \forall \, i,j \in \{1,\dots,2n\}\).        
\endproof

\begin{rem}
\upshape Theorem \ref{thm1} is originally proposed for systems whose coefficients \(a_{ij}(t)\)	are explicit functions of time \(t\). However, it has been shown in \cite{ref12} that the same result also applies to systems in which the coefficients depend on the state \(x\). For details, see \cite{ref12}.
\end{rem}

Recall from Theorem~\ref{thm2} that when \(\mathcal{G}\big(P(t)\big)\) has more than one SCC, its condensation digraph consists of a single sink and multiple singleton sources. Following the notation introduced in the proof of Theorem~\ref{thm2}, let \(\Omega(t)\) and \(\Theta(t)\) denote, respectively, the set of nodes forming the unique sink and the collection of singleton sources at time \(t\). Under this notation, the following proposition establishes convergence for Case~(ii) of Theorem~\ref{thm2}. 

\begin{prop} \label{prop3}
Suppose that there exists a finite time instant \(T>0\) such that for all \(t \geq T\), \(\Omega(t)=\Omega:=\Omega(T)\) and 
\(\Theta(t)=\Theta:=\Theta(T)\). Then \(z_i(t)=z_i(T), \, \forall \, i \in \Theta, \, \forall \, t \geq T,\) and \(\lim_{t \to \infty}\mathrm{dist}(z_j(t), \operatorname{conv}\{z_i(T) \!\mid\! i \in \Theta \})=0, \, \forall \, j \in \Omega\), i.e., the states of nodes in \(\Omega\) approach the convex hull spanned by the constant states of nodes in \(\Theta\). 
\end{prop}  

\proof
It is apparent that \(\mathcal{G}\big(P(t)\big)\) is no longer cut-balanced since each opinion node \(i \in \Theta\) only has outgoing edges to other nodes but no incoming edges from them, making Theorem~\ref{thm1} not applicable. Hence, we analyze convergence for the nodes in \(\Omega\) and \(\Theta\) separately.    

From Theorem~\ref{thm2}, if \(i \in \Theta(t)\), then \(z_i(t+1)=z_i(t)\). Under the assumption that the node membership of each SCC of \(\mathcal{G}\big(P(t)\big)\) remains fixed after time \(T\), i.e., \(\Omega(t)=\Omega:=\Omega(T), \, \Theta(t)=\Theta:=\Theta(T),\, \forall \, t \geq T\), it follows that \(z_i(t)=z_i(T)\), \(\forall \, i \in \Theta, \, \forall \, t \geq T\).

By applying a suitable permutation to the state vector \(\mathbf{z}\), we obtain \(\hat{\mathbf{z}} = [\hat{\mathbf{z}}_1, \hat{\mathbf{z}}_2]^\top \), where \(\hat{\mathbf{z}}_1\) and \(\hat{\mathbf{z}}_2\) correspond to the node indices in \(\Omega\) and \(\Theta\), respectively. 

Then, the discrete-time system \eqref{eq:opinion-action-z} can be rewritten as
\begin{equation} \label{e19}
\hat{\mathbf{z}}(t+1)=
\begin{bmatrix}
\hat{\mathbf{z}}_1(t+1) \\[2pt]
\hat{\mathbf{z}}_2(t+1)
\end{bmatrix}=
\begin{bmatrix}
\hat{P}_1(t) & \hat{P}_2(t) \\[2pt]
0 & I
\end{bmatrix}
\begin{bmatrix}
\hat{\mathbf{z}}_1(t) \\[2pt]
\hat{\mathbf{z}}_2(t)
\end{bmatrix}, \, \forall \, t \geq T,
\end{equation}
where \(I\) and \(0\) denote the identity matrix and zero matrix with dimension \(|\Theta| \times |\Theta|\) and \(|\Theta| \times |\Omega|\), respectively. The submatrix \(\hat{P}_1(t) \in \mathbb{R}^{|\Omega| \times |\Omega|}\) is row-substochastic, but each row of the combined block \([\hat{P}_1(t) \ \hat{P}_2(t) ]\) still sums to 1.  

The system \eqref{e19} can be viewed as a distributed containment control problem with multiple stationary leaders, which has been extensively studied in the literature \cite{ref13, ref18, ref19}. Each opinion node \(i \in \Theta\) acts as a stationary leader since \(\hat{\mathbf{z}}_2(t) = \hat{\mathbf{z}}_2(T)\) for all \(t \geq T\), while the nodes \(j \in \Omega\) serve as followers whose states evolve under the influence of leaders. In addition, since \(\hat{\mathbf{z}}\) is merely a permutation of \(\mathbf{z}\), the digraph \(\mathcal{G}\big(\hat{P}(t)\big)\) associated with the state matrix \(\hat{P}(t)\), defined in \eqref{e19}, is identical to \(\mathcal{G}\big(P(t)\big)\), which remains fixed within one time step  \(t \in [t_i,t_i+1)\) and switches at the next time instant \(t_i+1\), where \(t_i \in \mathbb{Z}_{>0}\).

According to Theorem~5.3 in \cite{ref13}, if there exists a finite integer \(m>0\) such that, for any \(t \in \mathbb{Z}_{>0}\), every follower is reachable from at least one leader in the union graph \(\mathcal{G}[A(t)+\dots+A(t+m-1)]\), where \(A(t)\) denotes the adjacency matrix of the interaction graph at time \(t\), then all followers asymptotically converge to the convex hull spanned by the leaders' constant values. In the special case where there is only one leader, all followers converge to its constant value and reach a consensus, as proved in Lemma~9 of \cite{ref20}.

For our system, as shown in the proof of Theorem~2, each opinion node \(i \in \Theta\) has a directed edge to its corresponding action node \(n+i \in \Omega\), and the subgraph induced by \(\Omega\) is always strongly connected. As a result, every follower \(j \in \Omega\) is reachable from every leader \(i \in \Theta\) in \(\mathcal{G}\big(P(t)\big)\) for all \(t \geq T\). Hence, in the union graph \(\mathcal{G}[P(t)+\dots+P(t+m-1)]\), every follower remains reachable from every leader for all \(t \geq T\) and \(m>0\). Thus, our system satisfies the conditions of Theorem~5.3 in \cite{ref13}, implying that the states of nodes in \(\Omega\) ultimately converge to the convex hull spanned by the constant state values of nodes in \(\Theta\), i.e., \(\lim_{t \to \infty}\mathrm{dist}(z_j(t),\operatorname{conv}\{z_i(T) \!\mid\! i \in \Theta \})=0, \, \forall \, j \in \Omega\).
\endproof

Based on Propositions~\ref{prop2} and~\ref{prop3}, the following theorem presents a general convergence result of the proposed system.

\begin{thm} \label{thm3}
Suppose that there exists a finite time instant \(T \in \mathbb{Z}_{> 0}\) after which the structure of the weighted digraph \(\mathcal{G}\big(P(t)\big)\) stabilizes, i.e., it has a fixed number of SCCs whose node memberships remain unchanged for all \(t \geq T\). Then, under \(\phi \in (0,1)\), the proposed system~\eqref{eq:opinion-action-z} converges to one of the following steady states:
\begin{enumerate}
\item[(i)] Consensus: if \(\mathcal{G}\big(P(t)\big)\) remains strongly connected;
\item[(ii)] Clustering where the states of nodes in the sink SCC approach the convex hull spanned by the constant state values of singleton source SCCs: if the condensation digraph of \(\mathcal{G}\big(P(t)\big)\) contains one sink and multiple singleton sources. 
\end{enumerate} 
\end{thm}

\proof
Case~(i) follows from Proposition~\ref{prop2} and Case~(ii) follows from Proposition~\ref{prop3}.
\endproof

\begin{rem} \label{remark2}
\upshape The convergence results of Theorem~\ref{thm3} rely on the assumption that there exists a finite time instant beyond which the structure of digraph \(\mathcal{G}\big(P(t)\big)\) stabilizes. A rigorous proof of this property is not provided here due to its analytical complexity and is left for future work. Nevertheless, numerical simulations show that this assumption holds consistently, as illustrated in the next section.  
\end{rem}     

\begin{rem} \label{remark3}
\upshape The convergence results of Theorem~\ref{thm3} do not rely on agents' initial conditions \(x_i(0)\) and \(y_i(0)\), because they only determine the value of the state vector at \(t=1\), namely, \(\mathbf{z}(1) = [\, x_1(1), \cdots, x_n(1), y_1(0), \cdots, y_n(0) \,]^\top\), and the earlier proofs do not explicitly depend on \(\mathbf{z}(1)\). 
\end{rem}

\section{Numerical Simulations} \label{sec5}
In this section, we present simulation results of two representative cases corresponding to the scenarios stated in Propositions~\ref{prop2} and~\ref{prop3}, respectively. The number of agents is set to \(n=|\mathcal{V}|=10\) in the simulation, resulting in \(2n=20\) nodes in the digraph \(\mathcal{G}\big(P(t)\big)\), where indices \(\{1,\dots,10\}\) and \(\{11,\dots,20\}\) correspond to the opinion states \(x_i\) and action states \(y_i\), respectively, as defined in \eqref{e6}. Each agent's opinion is initialized as \(x_i(0) \sim U[0,1]\) for all \(i \in \mathcal{V}\), where \(U[0,1]\) denotes the uniform distribution on \([0,1]\). The initial action is set equal to the initial opinion, i.e., \(y_i(0)=x_i(0),\, \forall \,i \in \mathcal{V}\).   

From the opinion updating rule \eqref{e1} and Theorem~\ref{thm2}, a large confidence threshold \(\epsilon\) increases the possibility that each agent has at least one neighbor at time \(t\), making the digraph \(\mathcal{G}\big(P(t)\big)\) more likely to be strongly connected. In contrast, when \(\epsilon\) is small, agents are more likely to be isolated, leading to a digraph \(\mathcal{G}\big(P(t)\big)\) composed of multiple SCCs. Figures~\ref{fig1} and~\ref{fig2} illustrate a representative example for each case separately. 

As shown in Fig.~\ref{fig1}, when \((\epsilon,\phi)=(0.3,0.5)\), all agents' opinions and actions converge to a common value at \(T_1 = 5\), indicating that system~\eqref{eq:opinion-action-z} achieves consensus. The corresponding digraph \(\mathcal{G}\big(P(t)\big)\) remains strongly connected for all \(t \geq T_1\), thereby satisfying the sufficient condition for consensus stated in Proposition~\ref{prop2}. 

In contrast, when \((\epsilon,\phi)=(0.05,0.5)\), as illustrated in Fig.~\ref{fig2}, the agents' opinions and actions split into several clusters and the corresponding digraph \(\mathcal{G}\big(P(t)\big)\) at steady state consists of multiple SCCs. The structure of digraph \(\mathcal{G}\big(P(t)\big)\) stabilizes after \(T_2=4\), satisfying the sufficient condition of Proposition~\ref{prop3}. In this case, nodes 2, 5 and 6 (highlighted in red in Fig.~\ref{fig2b}) have no incoming edges from other nodes, serving as leaders and forming three singleton sources in the condensation digraph of \(\mathcal{G}\big(P(t)\big)\), while the remaining nodes serve as followers and constitute a single sink in the condensation digraph. Furthermore, as shown in Fig.~\ref{fig2c}, the steady-state values of these leaders contain the two extremes among all nodes, where node 5 takes the maximum and node 6 takes the minimum. Thus, the steady-state values of nodes in the sink SCC lie within the convex hull defined by the constant state values of singleton source SCCs, consistent with Proposition~\ref{prop3}.

In summary, the two examples confirm the theoretical results of Theorem~\ref{thm3}, illustrating that the digraph structure stabilizes within finite time and that the system either converges to consensus or exhibits multiple opinion--action clusters in which some opinion nodes act as stationary leaders and other nodes are eventually contained in the convex hull of the leaders.

\begin{figure}[!t]
\centering
\subfloat[]{\includegraphics[width=0.48\columnwidth]{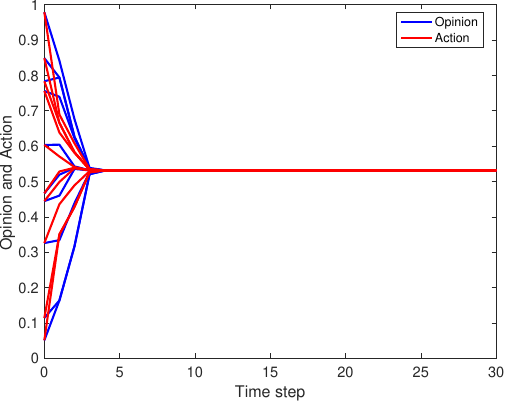}}
\hfill
\subfloat[]{\includegraphics[width=0.48\columnwidth]{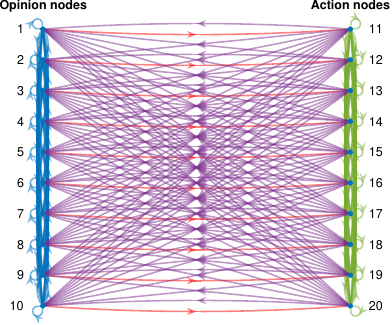}} 
\caption{Simulation results for \((\epsilon,\phi)=(0.3,0.5)\): (a) Evolution of \(x_i(t)\) and \(y_i(t)\), where the blue and red curves represent the trajectories of \(x_i(t)\) and \(y_i(t)\), respectively. (b) Steady-state structure of digraph \(\mathcal{G}\big(P(t)\big)\), where the blue, red, purple, and green curves denote opinion--opinion, opinion--action, action--opinion, and action--action edges, respectively.}
\label{fig1}
\end{figure}
\begin{figure}[!t]
\centering
\subfloat[]{\includegraphics[width=0.6\columnwidth]{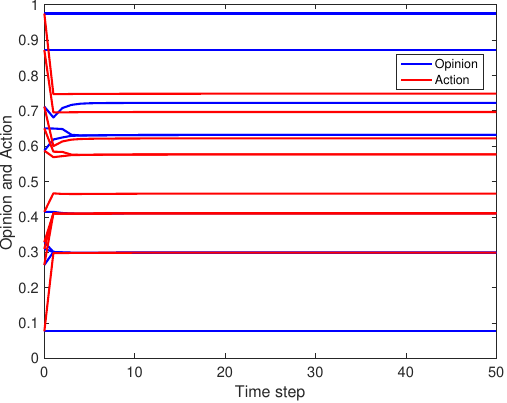}}  \\
\vspace{-5pt}
\subfloat[]{\includegraphics[width=0.48\columnwidth]{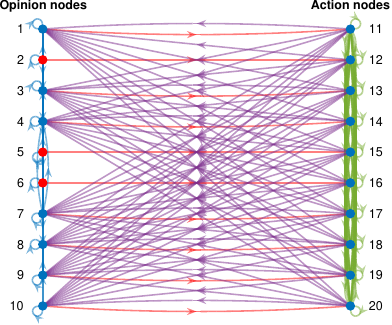}
\label{fig2b}} 
\hfill
\subfloat[]{\includegraphics[width=0.48\columnwidth]{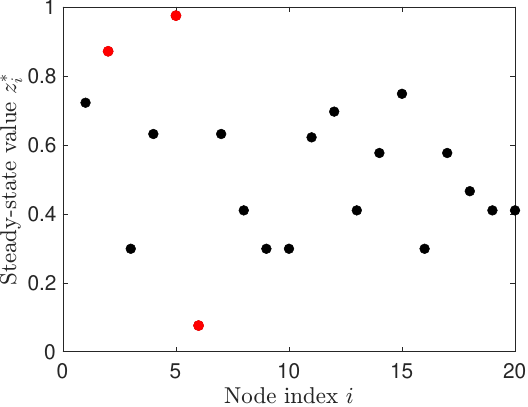}
\label{fig2c}} 
\caption{Simulation results for \((\epsilon,\phi)=(0.05,0.5)\): (a) Evolution of \(x_i(t)\) and \(y_i(t)\), where the blue and red curves represent the trajectories of \(x_i(t)\) and \(y_i(t)\), respectively. (b) Steady-state structure of digraph \(\mathcal{G}\big(P(t)\big)\), where the blue, red, purple, and green curves denote opinion--opinion, opinion--action, action--opinion, and action--action edges, respectively. Opinion nodes highlighted in red have no incoming edges from other nodes and act as leaders. (c) Steady-state node values, where the red points correspond to the leader nodes shown in panel~(b).}
\label{fig2}
\end{figure}

\section{Conclusions and Future Work} \label{sec6}
In this paper, we analyzed the convergence properties of the opinion--action coevolution model proposed in \cite{ref11}. The results show that the model’s steady-state behavior is primarily governed by the agents’ homogeneous confidence threshold \(\epsilon\). When \(\epsilon\) is large, agents are more susceptible to different opinions influenced by their neighbors, leading the entire population toward consensus. Conversely, when \(\epsilon\) is small, some agents become socially isolated and their opinions remain fixed, effectively acting as leaders that influence the evolution of other agents’ opinions and all actions. These analytical findings provide a solid theoretical explanation for the model’s steady-state behavior and its relevance to social dynamics. A promising direction for future work is to identify the threshold of \(\epsilon\) and \(\phi\) dividing the two convergence scenarios through sensitivity analysis. 

\addtolength{\textheight}{-12cm}   






\end{document}